%% file: main.tex
\documentclass{article}
\usepackage{graphicx} 
\usepackage{xcolor}

\newcommand{\corr}[1]{\textcolor{black}{#1}}

\title{\corr{A} synthetic Gannon--Lee incompleteness theorem\footnote{This is a condensed version of CR's Master semester project, written at EPFL in Spring 2025 under MB's supervision.}}
\author{Mathias Braun\thanks{Institute of Mathematics, EPFL, 1015 Lausanne, Switzerland. mathias.braun@epfl.ch}, Carlo Rotolo\thanks{\corr{Dipartimento di Matematica, Università di Pisa, 56127, Pisa, Italy. c.rotolo1@studenti.unipi.it}}}
\date{}

\input{preamble}

\begin{document}

\maketitle

\begin{abstract}
    We prove the Gannon--Lee incompleteness theorem for globally hyperbolic spacetimes. We assume the synthetic null energy condition of Ketterer and a trappedness condition we call   ``synthetically asymptotically regular''. Our result generalizes this classical result to the weighted case. It also motivates and  indicates extensions to low regularity, which are deferred to future work.  

\end{abstract}

\tableofcontents


\section{Introduction}

\renewcommand\footnotemark{}
\renewcommand\footnoterule{}
In this work we state and prove a version in a synthetic setting of a classical result in general relativity, the Gannon--Lee \corr{incompleteness} theorem. In \cite{ket}, Ketterer introduced synthetic equivalents to the notions of null energy condition and trapped surface. Using these synthetic characterizations we are able to modify the hypotheses of the classical Gannon--Lee theorem and recover the original conclusion under the new hypotheses. \corr{Our contribution lies in giving direct proofs that do not rely on the respective equivalences from \cite{ket}.} \corr{A new technical ingredient we establish is} a local to global result for Ketterer's synthetic NEC, which is of independent interest and is used in the last step of the proof.
We work with a weighted spacetimes $(M,g,e^{-V}\mathrm{Vol}_M),$ for which Ketterer's conditions are still defined and equivalent to opportune modifications of the curvature conditions of the classical Gannon--Lee theorem. Therefore, we also obtain an extension of the classical Gannon--Lee theorem for weighted spacetimes. Here we describe some background on \corr{incompleteness} theorems and synthetic curvature bounds and outline the content and main results of this work. 
\subsection{The \corr{incompleteness} theorems}
The \corr{incompleteness} theorems of general relativity are a collection of results that relate the geometry or topology of spacetime with the presence of ``singularities", in the sense of incomplete causal geodesics. This \corr{notion} is due to Penrose, who in \cite{penr} proved the first of these theorems. In his work he showed that in a spacetime that satisfies the \emph{null energy condition} (NEC), i.e. positivity of the Ricci curvature along null directions, and that contains a non-compact Cauchy hypersurface, the presence of a compact \emph{trapped surface} \emph{implies} the existence of an incomplete null geodesic. The concept of trapped surface was also introduced by Penrose and intuitively describes the boundary of a region of space of such strong gravity that not even light can escape it \corr{(which is made rigorous by the conclusion of his incompleteness theorem)}. The importance of Penrose's work is testified by the fact that he was awarded the Nobel Prize in Physics for it in 2020.
After him, many others proved different \corr{incompleteness} theorems: in \cite{hawk} Hawking applied a similar reasoning to the timelike setting to obtain existence of a past incomplete timelike geodesic and in \cite{hawk_penr} Hawking and Penrose were able to collect and strengthen many developments on \corr{incompleteness}  theorems in a more \corr{universal} result, known as the Hawking--Penrose \corr{incompleteness} theorem. Nowadays these results constitute both a central part of the classical theory of mathematical general relativity and an active area of research. For a recent review of this matter we cite the surveys of Senovilla and Garfinkle \cite{sen_garf} and Steinbauer \cite{stein}. 
\vskip 5pt
In this project we are interested in the Gannon--Lee theorem. \corr{In its classical form,} null geodesic incompleteness is proved assuming the null energy condition and some topological constraints on a codimension $2$ submanifold $S$ contained in a hypersurface $\Sigma.$ In its first version the Gannon--Lee theorem was proved independently by Gannon in \cite{gan_1} \cite{gan_2} and Lee in \cite{lee}. However their proofs contained a mistake that was found by Galloway, who also proved a corrected version of the theorem in \cite{gall}, strengthening the hypotheses. A modern version of the Gannon--Lee theorem, valid in higher dimension and with weaker \corr{hypotheses} on the causal structure of spacetime and the topology of $S$ was then proved by Costa e Silva \cite{ces}, but once again the proof relied on a false assumption on the causality of covering spacetimes. A correct proof was then provided by Costa e Silva and Minguzzi \cite{ces_ming} and Schinnerl and Steinbauer in \cite{schinnerl_steinbauer}, requiring \emph{past reflectivity} (see Definition \ref{past_refl}) of spacetime and some of its coverings. The main assumptions of this final version of the Gannon--Lee theorem are the null energy condition and the existence of an \emph{asymptotically regular hypersurface}. For this last hypotheses we require for a spacelike partial Cauchy surface $\Sigma$ to contain a submanifold $S$, of codimension $2$ in $M$, that separates it in two connected components $\Sigma_-$ and $\Sigma_+$. Moreover, we assume a kind of trappedness for $S$. However, in contrast with Penrose's theorem, we require that $S$ is trapped only in the direction of null geodesics pointing towards $\Sigma_-$, while in the direction of $\Sigma_+$ we give a topological constraint and require that the map $\pi_1(S)\rightarrow\pi_1(\Sigma_+\cup S)$, induced by the inclusion, is surjective.

\subsection{Synthetic curvature bounds}
Both the null energy condition and the existence of an asymptotically regular hypersurface are assumptions that concern to some extent the curvature of the spacetime $(M,g, e^{-V}\mathrm{Vol}_M)$ and therefore require the metric $g$ to be of regularity at least $C^2$. However, the interest in developing the theory of general relativity and Lorentzian geometry in lower regularity has grown very much recently, especially in relation to the \corr{incompleteness} theorems. \corr{For instance, we refer to the review of Cavalletti and Mondino \cite{review_cav_mond}.}
One possible approach for recovering these results in spacetimes with a low regularity metric is to regard the Riemann curvature tensor as a distribution and formulate conditions such as the NEC in the distributional sense. In \cite{schinnerl_steinbauer}, Schinnerl and Steinbauer used this method to extend the validity of the Gannon--Lee theorem to metrics of regularity $C^1$. One of the earliest reference for the distributional approach to curvature is \cite{distr} by Geroch and Traschen. For a modern explanation, together with some background on the relevance of the low regularity setting, we once again refer to the survey \cite{stein} by Steinbauer.
\vskip 5pt
Another possibility is to avoid any explicit mention of curvature, replacing conditions such as the NEC with other assumptions that can be formulated in spacetimes of low regularity or even in the more general setting of Lorentzian length spaces as developed by Kunzinger and Sämann in \cite{kunz_sam}. This is motivated by the great success of this approach in the Riemannian case:  in the initial works of McCann \cite{mccann_riem}, Otto--Villani \cite{ott_vill}, Cordero-Erausquin--McCann--Schmuckenschläger \cite{cord_mcc_schm} and von Renesse--Sturm \cite{vonren_sturm},
techniques from optimal transport were applied to derive characterization for Ricci curvature bounds on Riemannian manifolds via convexity of appropriate functionals along geodesics on the space of probability measures on $M$. This lead to a strong theory of Ricci curvature bounds for non-smooth metric space by Sturm \cite{sturm_1} \cite{sturm_2} and Lott and Villani \cite{lott_vill}. In the Lorentzian setting this is a new but rapidly growing area of research: in \cite{mccann_lor} McCann gave a characterization of the strong energy condition for weighted spacetimes, i.e., positivity of the $N$-Bakry--Emery--Ricci curvature in time-like directions, via convexity properties of the Boltzmann--Shannon entropy functional
\[ E_V(\mu) = \int_M \rho \log (\rho)\; e^{-V} \mathrm{dVol}_M\]
on the space of absolutely continuous probability measures $\mu$ with density $\rho$ on $M$ w.r.t. the weighted volume $e^{-V}$Vol$_M$. Following a similar approach Mondino and Suhr were able to derive a synthetic characterization of the Einstein equations of general relativity in \cite{mond_suhr}. This line of work was expanded by Cavalletti and Mondino in \cite{cav_mond}, where they were also able to extend the notions of synthetic curvature bounds to the non-smooth setting of Lorentzian length spaces. Other recent developments were obtained by Braun \cite{braun} and Braun--Ohta \cite{braun_ohta}, who also analyzed the relation between (Bakry--Emery) Ricci curvature and another entropy functional, the \emph{N-Renyi entropy}
\[S_N(\mu) =-\int_M \rho^{1-\frac{1}{N}}e^{-V}\mathrm{dVol}_M.\] 

The analogous development in the case of curvature bounds along null directions is even more recent. In this work we focus on the approach of Ketterer \cite{ket}, who gave the first characterization of the NEC via optimal transport along null geodesics. Despite being synthetic in nature, Ketterer's synthetic NEC requires $C^2$ regularity of the metric to be defined and it remains to be understood wheter it is possible to extend it to a non-smooth setting. \corr{Right after \cite{ket}, another proposal for a synthetic null energy condition was given by McCann in \cite{mccann_null} as a limit of synthetic timelike curvature bounds.} More recently, in \cite{cav_mond_man}, Cavalletti, Manini and Mondino were able to formulate yet another characterization of the NEC, working in the space of absolutely continuous probability measures on \corr{null hypersurfaces} w.r.t. fixed \corr{reference measures}, called the \corr{\emph{rigged measure}}, and in \cite{cav_mond_man_due} they extended their method to Lorentzian length space. \corr{Just like a first synthetic version of the Penrose incompleteness theorem was established  \cite{cav_mond_man_due}, we believe their framework also allows for a synthetic version of the Gannon--Lee incompleteness theorem; we will defer this to future work.}

\subsection{Content of the work and main results}
In \cite{ket} Ketterer introduced his synthetic counterparts to the null energy condition and to the notion of trapped surface. As an application, he was able to prove a Penrose \corr{incompleteness} theorem \corr{by synthetic means, without relying on classical curvature quantities}. Motivated by this, we apply the same approach to the Gannon--Lee theorem. 

Ketterer's synthetic null energy condition concerns the null geometry of spacetime. A \emph{null hypersurface} $\mathcal{H}$ in a weighted spacetime $(M,g,e^{-V}\mathrm{Vol}_M)$ is a $C^2$ hypersurface such that the restriction of the metric to $\mathcal{H}$ is degenerate. The metric induces a degenerate volume Vol$_\mathcal{H}$ on $\mathcal{H}$ and we define $m_\mathcal{H} = e^{-V}\mathrm{Vol}_\mathcal{H}.$ For an $m_\mathcal{H}$-absolutely continuous probability measure $\mu$ with density $\rho$ we consider the relative $N$-Renyi entropy
\[S_N(\mu\,\vert\,m_\mathcal{H}) = -\int_M\rho^{1-\frac{1}{N}} \mathrm{d}m_\mathcal{H}.\]

There exists a $C^1$ null normal vector field $K$ along $\mathcal{H}$, whose flow curves can be repara\-metrized to be null geodesics, which we call \emph{null geodesic generators} of $\mathcal{H}$. These geodesics foliate $\mathcal{H}$ and along them we can transport probability measures
supported on spacelike cross-sections of $\mathcal{H}$. Concretely, given two null connected (see Section \ref{snec} for a precise definition) probability measures $\mu_0$ and $\mu_1$, supported on spacelike cross-sections $S_0$ and $S_1$, \corr{respectively},  Ketterer showed that it is possible to transport $\mu_0$ to $\mu_1$ via a \emph{null displacement interpolation} $(\mu_t)_{t\in[0,1]}$ of probability measures. A weighted spacetime then satisfies the  \emph{synthetic} $N$-\emph{null energy condition} if for any two acausal probability measures that are null connected and are $m_\mathcal{H}$-absolutely continuous, the $N$-Renyi entropy is convex along the null displacement interpolation, i.e. \corr{for every $t\in [0,1]$,}
\[S_N(\mu_t\,\vert\,m_\mathcal{H}) \le (1-t)S_N(\mu_0\,\vert\,m_\mathcal{H}) + t S_N(\mu_1\,\vert\, m_\mathcal{H}).\]
In \cite{ket} Ketterer \corr{also gave} a synthetic characterization of trappedness of a codimension $2$ spacelike surface $S$ (see Definition \ref{synth_trapp}). He proved that $S$ is trapped if and only if the measures of geodesic balls in $S$ decreases quickly enough when transported along null geodesics exiting orthogonally to $S$.
\vskip 5pt
We can now state our synthetic Gannon--Lee theorem. The main notions required are the following:
\begin{itemize}
    \item past reflectivity of spacetime and some of its covering, as in the classical Gannon--Lee theorem (see Definition \ref{past_refl});
    \item the synthetic null energy condition (see above and Definition \ref{synth_nec});
    \item existence of a synthetically asymptotically regular hypersurface $\Sigma$, which we define by replacing the trappedness condition of Definition \ref{asymp_reg} with the same bounds required by Ketterer on the decrease of the measures of geodesic balls, while transported only along null geodesics pointing towards $\Sigma_-$ (see Definition \ref{synth_asymp_reg}). We also request that $\Sigma$ admits a piercing (see Section \ref{initial}).
\end{itemize}

\begin{theorem}[Synthetic Gannon--Lee theorem]
    Let $(M,g, e^{-V}\mathrm{Vol}_M)$ be a past reflecting, null geodesically complete weighted spacetime of dimension $n+1$ which satisfies the synthetic $N$-null energy condition for $N\ge n-1$. Let $\Sigma$ be a synthetically asymptotically regular hypersurface admitting a piercing, with enclosing surface $S.$ Suppose that also one of the two following possibilities holds:
    \begin{itemize}
        \item any covering spacetime of $(M,\,g,\, e^{-V}\mathrm{Vol}_M)$ is past reflecting, or
        \item S is simply connected and the universal covering spacetime of $(M,g, e^{-V}\mathrm{Vol}_M)$ is past-reflecting.
    \end{itemize}
    Then the map $i_\#:\pi_1(S)\rightarrow \pi_1(\Sigma)$, induced by the inclusion $ i:S\rightarrow \Sigma,$ is surjective.
\end{theorem}
\noindent We follow the structure of the proof presented by Costa e Silva and Minguzzi in \cite{ces_ming} and Schinnerl and Steinbauer in \cite{schinnerl_steinbauer}. The key result is compactness of the set $\overline{\Sigma}_-$, from which the thesis follows using a purely topological argument. However, the classical proof of this heavily relies on a focusing result for null geodesics exiting from $S$ in the direction of $\Sigma_-$ which requires both the NEC and inner trappedness of $S.$ We apply results from Ketterer \cite{ket} to obtain a similar focusing result in our synthetic setting and then we are able to adapt easily the other steps of the proof. Replacing past reflectivity with the stronger assumption of global hyperbolicity, we obtain the following easier statement of the theorem.
\begin{theorem}[Globally hyperbolic Gannon--Lee theorem]
     Let $(M,g, e^{-V}\mathrm{Vol}_M)$ be a globally hyperbolic, null geodesically complete weighted spacetime of dimension $n+1$ which satisfies the synthetic $N$-null energy condition for $N> n-1$. Let $\Sigma$ be a synthetically asymptotically regular Cauchy surface, with enclosing surface $S.$
     Then the map $i_\#:\pi_1(S)\rightarrow \pi_1(\Sigma)$, induced by the inclusion $ i:S\rightarrow \Sigma,$ is surjective.
\end{theorem}
The final part of the proof of the Gannon--Lee theorem, both in the classical and in the synthetic setting, consists in a topological argument involving covering spacetimes of $M$. In particular, it is crucial that all the hypotheses on $M$ lift to coverings. This is obvious for the NEC, since the Ricci curvature is a local object. Ketterer's synthetic NEC is equivalent to the classical one, hence it also has a local to global property. However, in continuity with our purpose of proving a truly synthetic Gannon--Lee theorem we prove this result independently in Section \ref{loc_glob}. \corr{This yields a first null version of the local-to-global result for the timelike curvature-dimension condition of Cavalletti and Mondino \cite{cav_mond} established by Braun \cite{braun}.} Our proof is inspired by many similar results obtained in the Riemaniann case, such as those proved by Sturm \cite{sturm_1}, Villani \cite{vill}, Bacher--Sturm \cite{BACHER201028}, Cavalletti--Millman \cite{Cavalletti2016TheGT}, and Ketterer \cite{ket_local_to_global}, and has also the advantage that it could be possibly adapted to some generalizations of the synthetic NEC to the non-smooth setting, and perhaps also to the one proposed in \cite{cav_mond_man_due} by Cavalletti, Manini and Mondino (which we defer to future work).
\subsection{Organization of the work}
The rest of this work is structured as follows: in section $2$ we briefly lay out the preliminary notions from Lorentzian geometry and then we introduce Ketterer's synthetic null energy condition, referring to his paper \cite{ket} for a deeper description. In section $3$ we state the Gannon--Lee theorem, both in its classical and synthetic version. We give the definition of past reflectivity, which is the causality hypotheses of the theorem, then we introduce rigorously the notion of asymptotically regular hypersurface and its synthetic counterpart. We conclude by stating the classical and synthetic Gannon--Lee theorem. In Section $4$ we stir away from the Gannon--Lee theorem and we prove our local to global result for the synthetic NEC. In Section $5$ we turn to the proof of the synthetic Gannon--Lee theorem. We include the proof of the focusing result from \cite{ket}, which is the main point where our proof differs from the classical one. Then we adapt the reasoning of Costa e Silva--Minguzzi \cite{ces_ming} and Schinnerl--Steinbauer \cite{schinnerl_steinbauer} for the latter part of the proof; some modifications are needed for compactness of $\overline{\Sigma}_-$, while the final topological step remains valid in our setting. 
\vskip 10pt
\noindent \textbf{Acknowledgments} MB is supported by the EPFL through a Bernoulli Instructorship. Most of this work was carried out when CR was visiting the EPFL for an exchange semester with the financial support of the Swiss-European Mobility Programme. CR also wishes to thank the INdAM research group GNSAGA for its support.
\section{Preliminary notions} We begin by laying out all the preliminary notions that are needed for understanding the statement and the proof of our results. In the first part of this section we will cover the basics of Lorentzian geometry and then we will turn to describe Ketterer's synthetic NEC. 
\subsection{Lorentzian geometry}
 Some references on Lorentzian geometry are the classical texts of Hawking and Ellis \cite{hawk_ell}, O'Neill \cite{oneill}, or the more recent \cite{ming_2} by Minguzzi. In the rest of this paper we will denote with $M$ a Lorentzian manifold of dimension $n+1$, i.e., a smooth manifold without boundary endowed with a metric tensor $g$ of signature $(-, +, \dots, +)$. We will require $g$ to be of regularity $C^2$, where all the classical theory of Lorentzian geometry extends without any difficulties (see Steinbauer \cite{stein}). We denote with $\langle v,w\rangle$ the bilinear form $g(v,w).$ A non-zero tangent vector $v\in TM$ is said to be \emph{timelike}, \emph{spacelike} or \emph{null} if $\langle v,v\rangle$ is respectively negative positive or zero. The zero vector is by convention spacelike. Similarly, a $C^1$ curve $\gamma : I\rightarrow M$ is \emph{timelike}, \emph{spacelike} or \emph{null} if its tangent vectors if so is its velocity $\gamma'(t)$ for all $t\in I$. \\
 \indent We say that $M$ is \emph{time oriented}, or a \emph{spacetime} if it admits a continuous non-vanishing vector field $Y$. A vector $v\in T_pM$ is said to be \emph{future directed} if $\langle Y(p), v\rangle < 0$ and \emph{past directed} if $\langle Y(p), v\rangle > 0$. \\
 \indent We say that a point $q$ lies in the \emph{chronological future} of $p$, and we write $p \ll q$, if there is a future directed timelike curve from $p$ to $q$. Analogously we say that $q$ lies in the \emph{causal future} of $p$, and we write $p\le q$, if there is a future directed causal curve from $p$ to $q.$ We then define the \emph{chronological} and \emph{causal future} of $p$ as 
$$I^+(p) = \{ q \in M\,\vert\, p \ll q\}, \qquad J^+(p)= \{q \in M\,\vert\, p\le q\}.$$
For a subset $A\subseteq M$ we define its chronological/causal future as the union of the chronological/causal future of all of its points. We denote these sets with $I^+(A), J^+(A).$ By considering the opposite relations we define also the \emph{chronological} and \emph{causal past} of $A$, denoted with $I^-(A)$ and $J^-(A).$ Some properties of the chronological and causal relations are the following (see Minguzzi \cite{ming_2}, Proposition 1.16):
\begin{proposition}[Properties of $\ll$ and $\le$]
    The relations $\ll$ and $\le$ are transitive. Moreover, the relation $\ll$ is open; in particular $I^{\pm}(S)$ is open for any subset $S\subseteq M$.
\end{proposition}
Transitivity can be strengthened with the so-called Push-up Lemma. For a proof see Theorem 2.24 of Minguzzi \cite{ming_2}.
\begin{proposition}[Push-up Lemma]
    Let $p,q,r$ be points in a spacetime $M$. The following statements hold:
    \begin{flalign*}
        &p\le q \ll r \Rightarrow p\ll r, \;and \\
        &p\ll q\le r \Rightarrow p\ll r.
    \end{flalign*}
\end{proposition}
Given an open subset $U$ containing $A$, we can also define the \emph{relative chronological/causal future/past of A in U}, by requesting the curves we consider to lie entirely in $U$. Concretely, we define 
$$I^+(A,U) = \{q\in U\,\vert\, \exists\;\mathrm{a\;future\;directed\; timelike\; curve\; from\; \emph{A}\; to \; \emph{q}\; contained \; in \; U}\}$$
and similarly for $I^-(A,U)$ and $J^{\pm}(A,U).$ \\
\indent We call a subset $A$ in a spacetime \emph{achronal} if it intersects at most once every timelike curve, and \emph{acausal} if it intersects at most once every causal curve. A \emph{partial Cauchy surface} is an achronal $C^2$ hypersurface in $M$. 
We call a $C^2$ hypersurface $\Sigma\in M$ a \emph{Cauchy surface} if every inextendible timelike curve intersects $\Sigma$ exactly once.
A spacetime admitting a Cauchy surface is called \emph{globally hyperbolic.} 

Given an achronal set $A$, its \emph{edge} is the set of points $p\in \overline{A}$ such that every open neighborhood $U$ of $p$ contains a timelike curve from $I^-(p, U)$ to $I^+(p,U)$ that does not intersect $A.$ A proof of the following can be found in O'Neill \cite{oneill}, Proposition 14.25.
\begin{proposition}[Achronal set with no edge is a topological hypersurface]
\label{achronal_edge}
An achronal subset $A\subseteq M$ is a topological hypersurface if and only if it is disjoint from its edge. 
\end{proposition}
\subsection{Ketterer's synthetic NEC}
\label{snec}
Now, we introduce the synthetic null energy condition proposed by Ketterer in \cite{ket}. For a spacetime $(M,g)$ with $g$ of regularity $C^2$, we say that the \emph{null energy condition holds} if 
\[\mathrm{Ric}(v,v) \ge 0\;\;\; \forall\;\mathrm{null\;vectors\;} v\in TM.\]
This assumption is a classical constraint on the Ricci curvature of spacetime that is considered physically reasonable and is one of the hypotheses of the Gannon--Lee theorem. In \cite{ket}, Ketterer proposed a synthetic condition that can be defined without any explicit mention of curvature and is equivalent to the NEC. Following his approach we will work in a slightly more general setting than the one in which the Gannon--Lee theorem is classically stated, that of weighted manifolds. \\
\indent A \emph{weighted spacetime} is a triple $(M,\,g,\, e^{-V}\mathrm{Vol}_M)$, where $V \in C^\infty(M)$ and Vol$_M$ is the volume form of $M.$ 
The \emph{Bakry--Emery Ricci tensor} of a weighted manifold $(M,\,g,\, e^{-V}\mathrm{Vol}_M)$ is given by 
$$\mathrm{Ric}^V := \mathrm{Ric} + \mathrm{Hess}\, V.$$
We say that a weighted spacetime satisfies the \emph{Bakry--Emery N-null energy condition} for $N > n-1$ if 
$$(N'-n+1)\mathrm{Ric}^V(v,v) \ge \langle \nabla V, v\rangle^2$$
for any null vector $v\in TM$ and for any $N' > N.$
\vskip 5pt
We now outline how to define Ketterer's equivalent to the Bakry--Emery $N$-null energy condition. For more details we refer to \cite{ket}. In a (weighted) spacetime, a \emph{null hypersurface} $\mathcal{H}$ is a $C^2$ hypersurface such that the restriction of the metric to $\mathcal{H}$ is degenerate. 
The metric tensor induces a degenerate $(n-1)$-volume form $\mathrm{Vol}_\mathcal{H}$ on $\mathcal{H}.$ For any Borel subset $A$ of a spacelike surface $S\subseteq \mathcal{H}$ of dimension $n-1$, this gives a definition of volume by imposing
$$\mathrm{Vol}_\mathcal{H}(A):=\int_A \mathrm{dVol}_\mathcal{H}.$$
If $V\in C^\infty(M)$ is the weight function of our spacetime, we define $m_\mathcal{H} := e^{-V}\mathrm{Vol}_\mathcal{H}$ and
we denote with $\mathcal{P}(M,m_\mathcal{H})$ the set of $m_\mathcal{H}$-absolutely continuous probability measures on $M$.  For a probability measure $\mu\in\mathcal{P}(M, m_\mathcal{H})$ with density $\rho$ we define its $N$\emph{-Renyi entropy}
$$ S(\mu\,\vert\, m_\mathcal{H}) := -\int_M \rho^{1-\frac{1}{N}}\mathrm{d}m_\mathcal{H}. $$

Given a point $p$ in a null hypersurface $\mathcal{H}$, there is a unique null direction in $T_p\mathcal{H}$ which is normal to $T_p\mathcal{H}.$ Time orientability of $M$ allows us to choose continuously a vector in this direction, giving rise to a null normal vector field $K$ along $\mathcal{H}.$ Locally we can reparametrize the flow curves of $K$ to make them null geodesics, which we call the \emph{null geodesic generators} of $\mathcal{H}.$ Since we assume that $\mathcal{H}$ has regularity $C^2$, it follows that $K$ will be of regularity at least $C^1$. From this, we can define a relation on $\mathcal{H}$ as follows:
$$R_\mathcal{H} := \{(x,y)\in \mathcal{H}^2\,\vert\, \exists\;\mathrm{a\; flow\; curve\;} \gamma\;\mathrm{of\;} K\; \mathrm{s.t.}\; \gamma(s) = x,\;\gamma(t) = y\; \mathrm{and}\; s\le t\}.$$
\begin{definition}[Null coupling]
    A \emph{null coupling} on $M$ is a probability $\pi \in \mathcal{P}(M^2)$ such that supp$(\pi)\subseteq R_\mathcal{H}$ for some null hypersurface $\mathcal{H}.$
\end{definition}
\begin{definition}[Null connected probability measures]
We say that two probability measures $\mu_0,\mu_1\in \mathcal{P}(M,m_\mathcal{H})$ are \emph{null connected} if there exists a null coupling $\pi$ with $\mu_0$ and $\mu_1$ as its marginals.
\end{definition}
For any $(x,y)\in R_\mathcal{H}$ we can consider the unique flow curve $\gamma_{x,y}$ of $K$ joining $x$ and $y$ and its affine reparametrization  $\tilde{\gamma}_{x,y} : [0,1]\rightarrow \mathcal{H}$. We denote with $\mathcal{G}(\mathcal{H})$ the set of all such affine reparametrized curves which inherits a structure of a Polish space.
\begin{definition}[Dynamical null coupling]
A \emph{dynamical null coupling} is a probability measure on $\mathcal{G}(\mathcal{H}).$
\end{definition}
If $\pi$ is a null coupling, the map $\Upsilon : R_\mathcal{H}\rightarrow \mathcal{G}(\mathcal{H})$ that sends $(x,y)$ to $\tilde{\gamma}_{x,y}$ is measurable and allows us to define the dynamical null coupling $\Pi := \Upsilon_\# \pi$. Now, if $e_t : \mathcal{G}(\mathcal{H})\rightarrow \mathcal{M}$ is the evaluation map at $t\in [0,1]$, we can define $\mu_t := (e_t)_\#\Pi.$ We call $(\mu_t)_{t\in[0,1]}$ the \emph{null displacement interpolation} induced by $\pi$ and say that $\pi$ is a null coupling between $\mu_0$ and $\mu_1.$ \vskip 5pt
We say that a probability measure is \emph{acausal} if its support is contained in an acausal spacelike $C^2$ submanifold. Given two null connected acausal probability measures $\mu_0,\mu_1\in P(M,m_\mathcal{H})$, with $\Sigma_0$ and $\Sigma_1$ spacelike hypersurfaces such that supp$(\mu_0)\subseteq \Sigma_0$
and supp$(\mu_1)\subseteq \Sigma_1$, it follows that actually supp$(\mu_0)\subseteq \mathcal{H}\,\cap\, \Sigma_0=:S_0$ and supp$(\mu_1)\subseteq \mathcal{H}\,\cap\, \Sigma_1=:S_1$. In this setting we have the two following crucial lemmata, both proved by Ketterer in \cite{ket}.
\begin{lemma}[Existence of a transport map, \cite{ket}, Lemma 3.2]
\label{transport_map}
Let $\mu_0,\mu_1,S_0,S_1$ be as above. Then there exists a $C^1$ function $r: U \rightarrow \R$, where $U$ is an open subset of $S_0$ that contains supp$(\mu_0)$, such that the map 
$$T:U\longrightarrow \mathcal{H},\qquad x\longmapsto \exp_x(r(x)K(x))$$
satisfies $\pi =(\mathrm{id}_{\mathrm{supp}(\mu_0)},T)_\# \mu_0.$ 
We call $T$ a transport map.
\end{lemma}
\begin{remark}
\label{diffeo_transp}
    Following the proof of the previous result, Ketterer also showed that we can extend the transport map $T$ to an entire family of maps $(T_t)_{t\in[0,1]},$ defined as follows: we set $\tilde{K}(x) = r(x)K(x)$ and define $T_t(x) = \exp_x(t\tilde{K}(x)).$ It holds that $\mu_t = (T_t)_\# \mu_0$ and that supp$\,\mu_t \subseteq T_t(S_0)=S_t$. Moreover the maps $T_t$ are diffeomorphisms from $S_0$ to $S_t$ for any $t\in[0,1],$ and $S_t$ is spacelike and acausal.
\end{remark}
\begin{lemma}[Change of variables, \cite{ket}, Lemma 3.6 and 3.7]
\label{area}
Let $\mu_0, \mu_1$ be $m_\mathcal{H}$-absolutely continuous acausal null connected probability measures and let $(\mu_t)_{t\in[0,1]}$ be the induced null displacement interpolation. Then each $\mu_t$ is absolutely continuous w.r.t. $m_\mathcal{H}$ and its density $\rho_t$ satisfies
$$e^{-V\circ T_t(x)} \rho_t(T_t(x)) \det DT_t(x) = \rho_0(x) e^{-V(x)}.$$
\end{lemma}
We can finally introduce Ketterer's synthetic null energy condition, which consists in requesting convexity of this entropy along the null displacement interpolation between any two acausal null connected probabilities in $\mathcal{P}(M,m_\mathcal{H}).$
\begin{definition}[Synthetic NEC]
\label{synth_nec}
    We say that a weighted spacetime $(M^{n+1},\, g,\, e^{-V}\mathrm{Vol}_M)$ satisfies the \emph{synthetic N-null energy condition} for $N > n-1$ if for every null hypersurface $\mathcal{H}$ and for any two acausal, null connected probability measures $\mu_0, \mu_1 \in \mathcal{P}(M, m_\mathcal{H})$, it holds 
    $$ S_{N'}(\mu_t\,\vert\, m_\mathcal{H}) \le (1-t) S_{N'}(\mu_0\,\vert\, m_\mathcal{H}) + t S_{N'}(\mu_1\,\vert\,m_\mathcal{H})\;\;\; \forall N'\ge N,$$
    where $\mu_t$ is the null displacement interpolation between $\mu_0$ and $\mu_1.$
\end{definition}
The following theorem is due to Ketterer and is proved in \cite{ket}.
\begin{theorem}[Equivalence of the NEC and the synthetic NEC]
A weighted spacetime $(M,g,e^{-V}\mathrm{Vol}_M)$ satisfies the Bakry-Emery $N$-null energy condition for $N > n-1$ if and only if it satisfies the synthetic $N$-null energy condition for $N > n-1$.
\end{theorem}
The synthetic NEC can be reformulated in terms of the densities $\rho_t$ instead of the entropy, thanks to the following lemma, also proved by Ketterer in \cite{ket}.
\begin{lemma}[Essential concavity]
\label{concavity}
Let $(M,\, g,\, e^{-V}\mathrm{Vol}_M)$ be a weighted time oriented spacetime of dimension $n+1$. Then $M$ satisfies the synthetic $N$-null energy condition for $N > n-1$ if and only if for any acausal probabilities $\mu_0, \mu_1 \in \mathcal{P}(M,m_\mathcal{H})$ that are null connected via a dynamical coupling $\Pi$ supported on a null hypersurface $\mathcal{H}$ it holds for any $t\in [0,1]$ 
$$\rho_t^{-\frac{1}{N}}(\tilde{\gamma}_t)\ge (1-t) \rho_0^{-\frac{1}{N}}(\tilde{\gamma}_0) + t \rho_1^{-\frac{1}{N}}(\tilde{\gamma}_1)\quad for\; \Pi\mathrm{-}almost\; every\; \tilde{\gamma}\in \mathcal{G}(\mathcal{H}),$$
where $\mu_t$ is the associated null displacement interpolation and $\rho_t$ is the density of $\mu_t$.
\end{lemma}
\section{Statement of the synthetic Gannon--Lee theorem}
\label{initial}
In this section we state both the classical and the synthetic version of the Gannon--Lee theorem, emphasizing the different conditions we are assuming in the formulation of the synthetic version. 
\vskip 5pt
We encode the theorem in the frame of the pattern \corr{incompleteness} theorem stated by Senovilla in \cite{sen}, which is a statement that outlines the general structures of many of the so called \corr{incompleteness} theorems, such as the theorems of Penrose \cite{penr}, Hawking \cite{hawk} and Hawking and Penrose \cite{hawk_penr}. In all these results the goal is to determine  conditions on a spacetime $M$ that are sufficient for it to be causally geodesically incomplete.  

\begin{theorem}[Pattern \corr{incompleteness} theorem]
    Let $(M,g)$ be a spacetime such that the following conditions hold:
    \begin{enumerate}
        \item a condition on the curvature of $M$;
        \item an initial or boundary condition;
        \item a causality condition.
    \end{enumerate}
    Then $(M,g)$ is causally geodesically incomplete, i.e., there exists a causal geodesic in $M$ whose maximal interval of definition is strictly contained in $\R.$
\end{theorem}
The energy condition for the Gannon--Lee theorem is the NEC, which we replace with the synthetic $N$-null energy condition. \\
\indent The causality condition is the same in the classical and in the synthetic version of the theorem, and is defined as follows:
\begin{definition}[Past reflecting spacetime]
\label{past_refl}
We say that a spacetime $M$ is \emph{past reflecting} if for any $p,q \in M$ one (hence both) of the following equivalent conditions holds: 
\begin{itemize}
    \item $I^+(p)\supseteq I^+(q) \Rightarrow I^-(p) \subseteq I^-(q);$
    \item $q\in \overline{I^+(p)}\Rightarrow p\in \overline{I^-(q)}.$
\end{itemize}
For a proof of the equivalence of the two conditions see for example Minguzzi \cite{ming_2}, Section 4.1.
\end{definition}
Now we describe the initial condition, which as the energy condition requires some modifications from the classical to the synthetic version. \\
\indent Consider a $C^2$ spacelike hypersurface $\Sigma \subseteq M$ in a weighted spacetime and a $C^2$ spacelike submanifold $S\subseteq \Sigma$ of codimension 2 in $M$. Suppose also that $S$ separates $\Sigma$, i.e., that $\Sigma\setminus S$ is the union of two disjoint open submanifolds $\Sigma_-$ and $\Sigma_+.$ This condition implies that the normal bundle of $S$ in $\Sigma$ is trivial and thus that we can choose normal vector fields $N_\pm$ along $S$ that are tangent to $\Sigma$ such that $N_+$ (the outer normal) points toward $\Sigma_+$ and $N_-$ (the inner normal) points toward $\Sigma_-.$ If we denote by $X$ the future directed timelike normal vector field to $\Sigma$ we get two null normal vector fields $K_\pm = \restr{X}{S} + N_\pm$ along $S$ that form at any $p\in S$ a null base of the normal space $T_pS^\perp\subseteq T_pM.$
\begin{definition}[Inner future converging surface, \cite{ket}, Definition 4.8]
\label{bak_em_trap}    With the above notation, we say that $S$ is \emph{Bakry--Emery inner trapped} or \emph{Bakry--Emery inner future converging} if for any $p\in S$ we have $$\langle H(p), K_-(p)\rangle + \langle \nabla V(p), K_-(p)\rangle>0, $$
    where $H(p)$ is the mean curvature vector of $S$.
    Analogously, using $K_+$, we may define an \emph{outer future converging} surface. 
\end{definition} 
Note that since $\Sigma$ and $S$ are of class $C^2$, the null normal vector fields $K_\pm$ are of regularity $C^1.$ Thus the mean curvature is defined classically and it is continuous. 
\begin{remark}
    When working in the non-weighted case the second term vanishes. That is the classical definition of inner/outer trapped surface.
\end{remark}

\begin{definition}[Asymptotically regular hypersurface, \cite{schinnerl_steinbauer}, Definition 2.5]
\label{asymp_reg}
    We say that a spacelike, connected partial Cauchy surface $\Sigma \subseteq M$ of regularity $C^2$ is an \textit{asymptotically regular hypersurface} if it contains a compact, connected, codimension $2$ submanifold $S$ such that:
    \begin{itemize}
        \item the complement $\Sigma\setminus S$ is the union of two disjoint submanifolds $\Sigma_+,\Sigma_-$ such that $\overline{\Sigma}_+$ is non-compact and $\Sigma_-$ is connected;
        \item the map $h_\# : \pi_1(S)\rightarrow \pi_1(\overline{\Sigma}_+)$, induced by the inclusion $h: S\rightarrow \overline{\Sigma}_+$, is surjective;
        \item $S$ is Bakry--Emery inner future converging.
    \end{itemize}
\end{definition}
Given an asymptotically regular hypersurface $\Sigma,$ we call a submanifold $S$ satisfying the above properties an \textit{enclosing surface}. Moreover, we say that $\Sigma$ admits a \textit{piercing} if there exists a timelike vector field such that all of its flow curves intersect $\Sigma$ exactly once. Given a piercing of $\Sigma,$ its flow curves foliates $M$ and we can define a continuous retraction $\rho_X : M\rightarrow \Sigma$ that sends a point $x\in M$ to the unique intersection of the flow curve passing through $x$ with $\Sigma.$
Now, for the initial condition of the classical Gannon--Lee theorem, we will assume that our spacetime $(M,g)$ contains an asymptotically regular hypersurface admitting a piercing.
\vskip 5pt
To modify the definition of asymptotically regular hypersurface we just need a synthetic analogue for the notion of future converging surface. Consider for each $p\in S$ the null geodesic $\gamma_p$ starting from $p$ with initial velocity $K_-(p)$ and the set $\mathcal{C} = \bigcup_{p\in S} \mathrm{Im}\,\gamma_p$. Then a neighborhood of $S$ in $\mathcal{C}$ is a null hypersurface that we denote with $\mathcal{H}$, and $K_-$ is the restriction to $S$ of its null normal vector field. Up to substituting $K_-$ with a suitable reparametrization $\Tilde{K}_-,$ we have that the map $$T_t : S\rightarrow \mathcal{H},\; p\mapsto \exp_p(t\Tilde{K}_-(p))$$ is a diffeomorphism onto its image for $t$ sufficiently small.

    \begin{definition}[Synthetically inner future converging surface, \cite{ket}, Definition 4.4]
    \label{synth_trapp}
    Using the notation introduced above, we say that $S$ is \textit{synthetically inner future converging}, or \textit{synthetically future converging in} $\mathcal{H}$, if for every $p \in S$ there exists $\epl(p)\in (0,1)$ such that

       $$ \restr{\frac{\mathrm{d}}{\underline{\mathrm{d}t}}}{t=0} \log m_\mathcal{H}(T_t(B_\delta(p))) \le \epl(p) <0, $$
    for any geodesic ball $B_\delta(p),$ where $\frac{\mathrm{d}}{\underline{\mathrm{d}t}}f(t) :=\lim\inf_{h\downarrow 0}\frac{f(t+h)-f(t)}{h}.$
    \end{definition}
The equivalence between this definition and the definition of Bakry--Emery inner future converging surface was also proved by Ketterer in \cite{ket}, Lemma 4.6.
\begin{theorem}
    Let $(M, g, e^{-V}\mathrm{Vol}_M)$ be a weighted spacetime and let $\mathcal{H}$ be a null hypersurface. Let $\Sigma$ be an acausal spacelike hypersurface and $S := \Sigma\,\cap\,\mathcal{H}.$ Then $S$ is synthetically future converging in $\mathcal{H}$ if and only if for any $p\in S$ it holds
    \[\langle H(p), K_-(p)\rangle + \langle \nabla V(p), K_-(p)\rangle >0,\]
    where $K_-$ is the future directed null normal of $\mathcal{H}.$
\end{theorem}
    Now we can easily give the definition of synthetically asymptotically regular hypersurface.
\begin{definition}[Synthetically asymptotically regular hypersurface]
\label{synth_asymp_reg}
    We say that a spacelike, connected partial Cauchy surface $\Sigma \subseteq M$ of regularity $C^2$ is a \textit{synthetically asymptotically regular hypersurface} if it contains a compact, connected, codimension $2$ submanifold $S$ such that:
    \begin{itemize}
        \item the complement $\Sigma\setminus S$ is the union of two disjoint submanifolds $\Sigma_+,\Sigma_-$ such that $\overline{\Sigma}_+$ is non-compact and $\Sigma_-$ is connected;
        \item the map $h_\# : \pi_1(S)\rightarrow \pi_1(\overline{\Sigma}_+)$, induced by the inclusion $h: S\rightarrow \overline{\Sigma}_+$, is surjective;
        \item $S$ is synthetically inner future converging.
    \end{itemize}
\end{definition}
As before we call $S$ an enclosing surface for $\Sigma$ and we define in the same way a piercing for $\Sigma.$
    \begin{remark}
        An enclosing surface, being contained in a $C^2$ spacelike partial Cauchy surface, is also an acausal set. This allows us to apply the synthetic $N$-null energy condition on probability measures with support in $S.$ 
    \end{remark}
Note that this initial condition doesn't impose only constraints on the geometry of $S, \Sigma $ and $M$, but also on their topology, thus exhibiting a strong connection between the topology of spacetime and null geodesic incompleteness.

We can now state both the classical and the synthetic Gannon--Lee theorem. We state the synthetic version in the setting of a weighted spacetime. Note that as proved by Ketterer in \cite{ket}, in a spacetime with a $C^2$ metric all the synthetic conditions we assume are equivalent to the classical ones, so the proof we present gives also a generalization of the classical Gannon--Lee Theorem to weighted spacetimes. 
\begin{theorem}[Classical Gannon--Lee Theorem, \cite{schinnerl_steinbauer}]
\label{class_gan}
    Let $(M,g)$ be a past reflecting, null geodesically complete spacetime which satisfies the null energy condition. Let $\Sigma$ be an asymptotically regular hypersurface admitting a piercing, with enclosing surface $S.$ Suppose that also one of the two following possibilities holds:
    \begin{itemize}
        \item any covering spacetime of $(M,g)$ is past reflecting, or
        \item S is simply connected and the universal covering spacetime of $(M,g)$ is past-reflecting.
    \end{itemize}
    Then the map $i_\#:\pi_1(S)\rightarrow \pi_1(\Sigma)$, induced by the inclusion $ i:S\rightarrow \Sigma,$ is surjective.
\end{theorem}
\begin{theorem}[Synthetic Gannon--Lee Theorem]
\label{synth_gan}
    Let $(M,g, e^{-V}\mathrm{Vol}_M)$ be a past reflecting, null geodesically complete weighted spacetime of dimension $n+1$ which satisfies the synthetic $N$-null energy condition for $N>n-1$. Let $\Sigma$ be a synthetically asymptotically regular hypersurface admitting a piercing, with enclosing surface $S.$ Suppose that also one of the two following possibilities holds:
    \begin{itemize}
        \item any covering spacetime of $(M,\,g,\, e^{-V}\mathrm{Vol}_M)$ is past reflecting, or
        \item S is simply connected and the universal covering spacetime of $(M,g, e^{-V}\mathrm{Vol}_M)$ is past-reflecting.
    \end{itemize}
    Then the map $i_\#:\pi_1(S)\rightarrow \pi_1(\Sigma)$, induced by the inclusion $ i:S\rightarrow \Sigma,$ is surjective.
\end{theorem}
Being globally hyperbolic is a stronger condition than being past reflecting (see Minguzzi \cite{ming_2}, Section 4.1 and 4.5) and it also passes to coverings. Thus, assuming global hyperbolicity of $M$ we get a simpler version of the statement, although slightly weaker. We state it below for completeness, however since the proof is basically the same with only some simplification in the causality step, we will present here only the general case.
\begin{theorem}[Globally hyperbolic Gannon--Lee Theorem]
     Let $(M,g, e^{-V}\mathrm{Vol}_M)$ be a globally hyperbolic, null geodesically complete weighted spacetime of dimension $n+1$ which satisfies the synthetic $N$-null energy condition for $N> n-1$. Let $\Sigma$ be a synthetically asymptotically regular Cauchy surface, with enclosing surface $S.$
     Then the map $i_\#:\pi_1(S)\rightarrow \pi_1(\Sigma)$, induced by the inclusion $ i:S\rightarrow \Sigma,$ is surjective.
\end{theorem}
In all the versions above we have stated the theorem directly, but we can easily obtain a result more similar to the pattern \corr{incompleteness} theorem discussed above by taking the converse. Concretely, this means removing the null geodesic completeness assumption and adding the hypothesis that $i_\# : \pi_1(S)\rightarrow\pi_1(\Sigma)$ is not surjective. This last assumption is an addition to the initial condition and can be interpreted as the assumption that topological complexities of $\Sigma$, in the form of non-trivial elements of its fundamental group, are not all coming from $S$. 

\section{A \corr{local to global} property for the synthetic NEC}
\label{loc_glob}
Before moving on to the proof of the synthetic Gannon--Lee theorem we prove a local to global result for Ketterer's synthetic NEC, that we will use in the final part of the proof of the theorem. In the smooth setting, which we are working in, this is a consequence of the equivalence with the classical NEC, since the Ricci tensor is a local object. However, our goal in this paper is to give a proof of the synthetic Gannon--Lee theorem that avoids any explicit mention of curvature and thus we will prove this result independently. 
This approach has also another advantage, since it could be adapted more easily to extend the validity of the result if the synthetic NEC were to be extended to a lower regularity setting.

\begin{definition}
    We say that the synthetic $N$-null energy condition for $N>n-1$ holds \emph{locally} in a weighted spacetime $(M, g, e^{-V}\mathrm{Vol}_M)$ of dimension $n+1$ if  
    for any null hypersurface $\mathcal{H}$, any point $p\in\mathcal{H}$ has an open neighbourhood $U$, with respect to the topology of $\mathcal{H}$, such that for any two acausal, null-connected, absolutely continuous probability measures $\mu_0,\mu_1\in \mathcal{P}(M,m_\mathcal{H})$ that are supported in $U$, the induced null displacement interpolation $(\mu_t)_{t\in[0,1]}$ satisfies
    \[ S_{N'}(\mu_t\,\vert\,m_\mathcal{H}) \le (1-t) S_{N'}(\mu_0\,\vert\, m_\mathcal{H}) + t S_{N'}(\mu_1\,\vert\, m_\mathcal{H})\]
     for all $N'\ge N$.
\end{definition}
We now prove that this is actually equivalent to the global synthetic NEC. There have been many works where local to global results for various synthetic curvature conditions were proved in the Riemannian setting, such as Sturm \cite{sturm_1}, Villani \cite{vill}, Bacher--Sturm \cite{BACHER201028},  Cavalletti--Milman \cite{Cavalletti2016TheGT}, Ketterer \cite{ket_local_to_global}, \corr{and their Lorentzian adaptation of Braun \cite{braun};} our result is very much inspired by these works. For the structure of the proof, we refer in particular to Ketterer \cite{ket_local_to_global}.
\begin{theorem}
\label{local_to_global}
    The synthetic $N$-null energy condition for a weighted spacetime \\ $(M, g, e^{-V}\mathrm{Vol}_M)$ holds if and only if it holds locally.
    \begin{proof}
        We start by proving that the global synthetic NEC holds for compactly supported probability measures and then we extend the result to the general case. Consider a null hypersurface $\mathcal{H}$ and two acausal, null connected probability measures $\mu_0,\mu_1\in\mathcal{P}(M,m_\mathcal{H})$ that are compactly supported. By Lemma \ref{concavity}, it suffices to show that for $\Pi-$a.e. $\gamma \in \mathcal{G}(\mathcal{H})$ it holds
         \[\rho_t(\gamma(t))^{-\frac{1}{N}} \ge (1-t) \rho_0(\gamma(0))^{-\frac{1}{N}} + t\rho_1(\gamma(1))^{-\frac{1}{N}},\]
         where $\Pi$ is the dynamical null coupling connectin $\mu_0$ and $\mu_1$ and the fuctions $\rho_t$'s are the densities of the induced null displacement interpolation $(\mu_t)_{t\in[0,1]}.$
        Since the null displacement interpolation is induced by the family of $C^1$ diffeomorphisms $T_t$, the set
        \[X=\bigcup_{t\in[0,1]} \mathrm{supp}\,\mu_t\]
        is compact. By Theorem 1 of Nomizu and Ozeki's paper \cite{compl_riem}, we can fix a complete background Riemannian metric $h$ on $M$. Using compactness, we can choose $\lambda \in (0,\mathrm{diam}\, X)$ and two finite collections of open subsets that cover $X$,
        \[ \mathcal{U}_1 = \{U_1,\dots, U_k\},\qquad \mathcal{U}_2 = \{V_1,\dots, V_k\},\]
        such that $B_\lambda(V_i)\subseteq U_i$ and the synthetic NEC holds for measures supported on $U_i$ for all $i = 1,\dots, k.$ 
        We refine the cover $\mathcal{U}_2$ by defining iteratively
        \[ L_1 = U_1, \qquad L_{i+1} = U_{i+1} \setminus \bigcup_{j = 1}^i \,(\,U_{i+1} \,\cap\, L_j\,).\]
        This way, we obtain a cover (not necessarily open) of $X$ with the disjoint, non-zero measure sets $L_1,\dots, L_k$. 
        \vskip 5pt
        Now, we can consider the maximal value of the $h$-length of the velocity of the null geodesics generator of $\mathcal{H}$ in $X$, which is a finite number thanks to compactness. Denote it with $m$ and take $\bar{t} < \bar{s}$ in $[0,1]$ such that $\bar{s}-\bar{t} \le \frac{\lambda}{m}.$ Define $\nu_0 = \mu_{\bar{t}}$ and $\nu_1 = \mu_{\bar{s}}$; these are acausal probability measures that are null connected by the family of transport maps $\widetilde{T}_\tau := T_{(1-\tau)\bar{t} + \tau \bar{s}}\circ  T_{\bar{t}}^{-1},$ for $\tau\in [0,1].$  We define 
        \[\nu_0^i := \frac{1}{\nu_0(L_i)} \restr{\nu_0}{L_i},\]
        so that 
        \[\nu_0 = \sum_{i=1}^k \nu_0(L_i) \nu_0^i.\] 
        If $\mathcal{L}_i = \{ \gamma \in \mathcal{G}(\mathcal{H})\,:\, \gamma(0)\in L_i\}$, the dynamical null couplings 
        \[\Pi^i = \frac{1}{\Pi(\mathcal{L}_i)} \restr{\Pi}{\mathcal{L}_i}\]
        induce for all $i = 1,\dots, k$ a null displacement interpolation $(\nu_\tau^i)_{\tau\in[0,1]}$ with transport maps given by $\restr{\widetilde{T}_\tau}{L_i}$. Using the fact that $\widetilde{T}_\tau$ is a diffeomorphism from the support of $\nu_0$ to the support of $\nu_\tau$ it is easy to check that
        \[\nu_\tau^i = \frac{1}{\nu_0(L_i)} \restr{\nu_\tau}{\widetilde{T}_\tau(L_i)},\]
        and therefore
        \[\nu_\tau = \sum_{i=1}^k \nu_0(L_i) \nu_\tau^i, \;\; \mathrm{and}\;\; \eta_\tau = \sum_{i=1}^k \nu_0(L_i) \eta_\tau^i,\]
        where $\eta_\tau$ is the density of $\nu_\tau$ and $\eta_\tau^i$ is the density of $\nu_\tau^i.$
        The support of $\nu_1^i$ is contained in $\widetilde{T}_1(L_i)$ and therefore is obtained from $L_i$ by moving across the segments of the null geodesic generators of $\mathcal{H}$ that are supported on $[\bar{t}, \bar{s}]$. Since by hypotheses $\bar{s} - \bar{t} \le \frac{\lambda}{m}$, these segments are shorter (with respect to $h$) than $\lambda$ and thus supp$\,\nu_1^i\subseteq B_\lambda(L_i)\subseteq U_i$. By Lemma \ref{concavity}, the local synthetic NEC then implies that for $\Pi^i-$a.e. $\gamma\in \mathcal{G}(\mathcal{H})$ it holds 
        \[\eta_\tau^i(\gamma(\tau))^{-\frac{1}{N}} \ge (1-\tau) \eta_0^i(\gamma(0))^{-\frac{1}{N}} + \tau\eta_1^i(\gamma(1))^{-\frac{1}{N}}.\]
        Since $\eta_\tau = \sum_{i=1}^k \nu_0(L_i)\eta_\tau^i$ and the supports of $\eta_\tau^i$ are pairwise disjoint, it also holds
        \[ \eta_\tau(x)^{-\frac{1}{N}} = \sum_{i=1}^k \nu_0(L_i)^{-\frac{1}{N}}\eta_\tau^i(x)^{-\frac{1}{N}}.\]
        Summing over $i$ and using the previous inequality we obtain for $\Pi -$a.e. $\gamma \in \mathcal{G}(\mathcal{H})$ 
        \[\eta_\tau(\gamma(\tau))^{-\frac{1}{N}} \ge (1-\tau) \sum_{i=1}^k \nu_0(L_i)^{-\frac{1}{N}}\eta_0^i(\gamma(0))^{-\frac{1}{N}} + \tau \sum_{i=1}^k \nu_0(L_i)^{-\frac{1}{N}}\eta_1^i(\gamma(1))^{-\frac{1}{N}}=\]
        \[ = (1-\tau) \eta_0(\gamma(0))^{-\frac{1}{N}} + \tau \eta_1(\gamma(1))^{-\frac{1}{N}}.\]
        This implies concavity of the function $(\bar{t}, \bar{s}) \ni \tau \mapsto \eta_\tau(\gamma(\tau))^{-\frac{1}{N}}$ and hence also $[\eta_\tau(\gamma(\tau))^{-\frac{1}{N}}]'' \le 0$ on $(\bar{t},\bar{s})$ in the distributional sense. We can reiterate the argument for different values of $\bar{t},\bar{s}$ and hence we obtain negativity of the second distributional derivative of $(0,1)\ni t\mapsto \rho_t(\gamma(t))^{-\frac{1}{N}},$ which is once again equivalent to
        \[\rho_t(\gamma(t))^{-\frac{1}{N}} \ge (1-t) \rho_0(\gamma(0))^{-\frac{1}{N}} + t\rho_1(\gamma(1))^{-\frac{1}{N}},\]
        as we wanted.
        \vskip 5pt
        Now we conclude the proof by analyzing the case of general probability measures that aren't necessarily compactly supported. 
        Let $\mu_0, \mu_1 \in \mathcal{P}(M,m_\mathcal{H})$ be acausal and null connected via a family of transport maps $(T_t)_{t\in[0,1]}$, and let $(\mu_t)_{t\in[0,1]}$ be the induced null displacement interpolation. We can consider an exhaustion by compact sets $\{B_n\}_{n\in\N}$ of $M$ and the restrictions 
        \[\mu_0^{(n)} = \frac{1}{\mu_0(B_n)}\restr{\mu_0}{B_n}.\]
        The restrictions of the maps $T_t$ induce for each $n\in\N$ a null displacement interpolation $\left(\mu_t^{(n)}\right)_{t\in[0,1]}.$ As before, using that $T_t$ is a diffeomorphism we obtain that for all $t\in [0,1]$ and for all $n\in \N$ it holds
        \[\mu_t^{(n)} = \frac{1}{\mu_0(B_n)} \restr{\mu_t}{T_t(B_n)}, \qquad \rho_t^{(n)} = \frac{1}{\mu_0(B_n)} \restr{\rho_t^{(n)}}{T_t(B_n)},\]
        where $\rho_t^{(n)}$ is the density of $\mu_t^{(n)}.$
        Now, $\mu_0^{(n)}$ and $\mu_1^{(n)}$ are compactly supported and therefore we have convexity of the $N$-Renyi entropy along the null displacement $(\mu_t^{(n)})$. By Lemma \ref{concavity}, it holds
        \[ \rho_t^{(n)}(\gamma(t))^{-\frac{1}{N}} \ge (1-t) \rho_0^{(n)}(\gamma(0)) + t \rho_1^{(n)}(\gamma(1))\]
        for $\Pi_n$-almost every curve $\gamma$ in $\mathcal{G}(\mathcal{H}),$ where $\Pi_n$ is the dynamical null coupling induced by $(\mu_t^{(n)})_{t\in[0,1]}$. Since $\rho_t^{(n)}$ converges pointwise to $\rho_t$ and the union of the supports of $\Pi_n$ is the support of $\Pi$, we can pass to the limit to obtain that
        \[\rho_t(\gamma(t))^{-\frac{1}{N}} \ge (1-t) \rho_0(\gamma(0))^{-\frac{1}{N}} + t\rho_1(\gamma(1))^{-\frac{1}{N}}\]
        for $\Pi$-almost every $\gamma\in\mathcal{G}(\mathcal{H})$. Once again by Lemma \ref{concavity} the thesis follows.
        \end{proof}        
\end{theorem}
\begin{corollary}[Synthetic NEC passes to coverings]
\label{covering}
If $\Phi : \tilde{M}\rightarrow M$ is a covering of spacetimes and $M$ satisfies the synthetic NEC then $\tilde{M}$ satisfies the synthetic NEC.
\begin{proof}
    This follows immediately from the previous theorem: since $\Phi$ is a local isometry, $\tilde{M}$ satisfies the local synthetic NEC and therefore also the global synthetic NEC.
\end{proof}
\end{corollary}
\section{Proof of the theorem}
We now prove the synthetic Gannon--Lee theorem. We follow the proof of Schinnerl and Steinbauer \cite{schinnerl_steinbauer}, modifying the step in which inner trappedness of $S$ and the NEC come into play. We instead substitute those results with the corresponding ones from Ketterer \cite{ket}. The structure of the proof is as follows: the first step is proving the existence of a focal point along the null geodesics orthogonal to $S$. This is the main point in which we need to use the synthetic condition instead of the classical ones. After this, we can basically follow the proof of Schinnerl and Steinbauer \cite{schinnerl_steinbauer}, with some simplifications due to the fact that we are working assuming higher regularity of the metric tensor. The second step of the proof is showing that the set $\overline{\Sigma}_-:= S\,\cup\,\Sigma_-$ is compact. Then, we can conclude with a purely topological argument that allows to obtain surjectivity of the map $i_\# : \pi_1(S)\rightarrow \pi_1(\Sigma).$
\subsection{Focusing result}
The key result of this step of the proof is proved by Ketterer in Remark 4.10 of \cite{ket}, but we include the proof for completeness, since it is a crucial step in the proof of the synthetic Gannon--Lee theorem and is one of the main points in which our proof differs from the classical one. 
\vskip 5pt Recall that given a compact, acausal, codimension $2$ submanifold $S\subseteq M$ endowed with a $C^1$ null vector field $K_-$ that is normal to $S$, we can consider for any $x\in S$ the geodesic $\gamma_x$ starting from $p$ with initial velocity $K_-(x)$ and we can construct a null hypersurface $\mathcal{H}$ by taking a neighborhood of $S$ in $C = \bigcup_{x\in S} \mathrm{Im}\,\gamma_x.$ By Remark \ref{diffeo_transp} we can choose a reparametrization $\tilde{K}_-$ of $K$ such that for sufficiently small $t$ the map 
\[T_t : S \rightarrow \mathcal{H}, \; x\mapsto \exp_x(t\tilde{K}_-(x))\]
is a diffeomorphism onto its image.
\begin{proposition}[Focusing result] 
\label{focusing}
    Let $(M, g, e^{-V}\mathrm{Vol}_M)$ be a weighted spacetime that satisfies the synthetic $N$-null energy condition and $S$ a compact, acausal, spacelike codimension $2$ submanifold. Consider $\mathcal{H}$, $\Tilde{K}_-$ and $T_t$ as above and suppose that $S$ is synthetically future converging in $\mathcal{H}.$ Then there exists a time $t_0>0$ such that for all $x\in S$ and for all geodesic balls $B_\delta(x),$ the map $\restr{T_t}{B_\delta(x)}$ ceases to be a diffeomorphism for some $t\in (0,t_0).$
\begin{proof}
    Consider $x\in S$ and a geodesic ball $B_\delta(x)$. Let $t_x>0$ such that the map $T_t$ is a diffeomorphism for $t\in (0,t_x).$ Consider the acausal probability measure $\mu_0 := \frac{1}{m_{\mathcal{H}}}m_\mathcal{H} \vert_{B_\delta(x)}$ and the induced null displacement interpolation $\mu_t := (T_t)_\# \mu_0.$ Finally, let $\rho_t$ be the density of $\mu_t$ with respect to $m_\mathcal{H}$. From Lemma \ref{area}, we have
    $$m_\mathcal{H}(T_t(B_\delta(x))) = \int_{B_\delta(x)} \det DT_t(x) e^{-V\circ T_t}\mathrm{dVol}_\mathcal{H} = \int_{B_\delta(x)} \frac{\rho_0}{\rho_t \circ T_t}\mathrm{d}m_\mathcal{H}.$$
    
    
    Since $S$ is future converging in $\mathcal{H}$ we get
    $$0 > -\epl(x) \ge \frac{1}{m_\mathcal{H}(B_\delta(x))}  \restr{\frac{\mathrm{d}}{\underline{\mathrm{d}t}}}{t=0}\int_{B_\delta(x)} \frac{\rho_0}{\rho_t\circ T_t} \mathrm{d}m_\mathcal{H} \ge\frac{1}{m_\mathcal{H}(B_\delta(x))}  \int_{B_\delta(x)} \restr{\frac{\mathrm{d}}{\underline{\mathrm{d}t}}}{t=0}\frac{\rho_0}{\rho_t\circ T_t} \mathrm{d}m_\mathcal{H},$$
    where in the last inequality we have used Fatou's lemma. 
    \vskip 5pt
    Letting $\delta \downarrow 0$ we obtain
    $$-\epl(x) \ge \restr{\frac{\mathrm{d}}{\underline{\mathrm{d}t}}}{t=0}\frac{\rho_0}{\rho_t\circ T_t} = -\restr{\frac{\rho_0}{(\rho_t \circ T_t)^2}}{t=0} (\rho_t \circ T_t)'(0) = -\frac{(\rho_t \circ T_t)'(0) }{\rho_0}.$$
    Now consider the function 
    $$y_x(t) = \left[ \log(\rho_t\circ T_t(x))\right]' = \frac{(\rho_t\circ T_t)'}{\rho_t\circ T_t}.$$
    From the previous computation we have $y_x(0) \ge \epl(x) > 0.$ Moreover,  Lemma \ref{concavity} implies concavity of the function $t\mapsto (\rho_t\circ T_t)^{-\frac{1}{N}}$. Setting for simplicity of notation $u(t) := \rho_t\circ T_t(x)$ we get
    $$\left[\left(\rho_t\circ T_t\right)^{-\frac{1}{N}}\right]''= \frac{N+1}{N^2}u^{-\frac{1}{N}-2}(u')^2-\frac{1}{N}u^{-\frac{1}{N}-1}u'' \le 0,$$
    which is equivalent to 
    $$u''u^{-1} - \frac{N+1}{N} \left(\frac{u'}{u}\right)^2 = \frac{u''u-(u')^2}{u^2} -\frac{1}{N}\left(\frac{u'}{u}\right)^2 = y_x'-\frac{1}{N}y_x^2 \ge0.$$
    Hence, by Riccati comparison, we obtain $t_x \le \frac{N}{\epl(x)}$ and to finish the proof it suffices to take $$t_0 := \frac{N}{\min_{x\in S}\epl(x)},$$ which is well defined by compactness of $S$.
\end{proof}
\end{proposition}
\begin{remark}
    The previous proposition shows that for any $x\in U$ the map $T_t$ is not a local diffeomorphism in $x$ for some $t\in [0, t_0].$ Hence its differential $DT_t(x)$ is singular, i.e., there exists a vector $v\in T_xS$ such that $DT_t(x)v = 0.$ Since the vector field $J(t) = DT_t(x)v$ along the null geodesic $\gamma(t)=\exp_x(t\Tilde{K}_-(x))$ is a Jacobi field, this means that there is a focal point along $\gamma$ at $t.$
\end{remark}

\subsection{Causality step}
The focusing result we have just shown is instrumental in the next step of the proof, which is showing the compactness of $\overline{\Sigma}_- = \Sigma_-\,\cup\, S.$ We adapt the proof found in Costa e Silva--Minguzzi \cite{ces_ming} and Schinnerl--Steinbauer \cite{schinnerl_steinbauer} to our synthetic setting. After the first step, in which we use Proposition \ref{focusing} to replace a focusing argument obtained via the classical NEC, the proof proceed analogously as in \cite{ces_ming} and \cite{schinnerl_steinbauer}, but we will include it anyways for convenience of the reader. We will use the following result on causal curves, which is theorem 10.51 in O'Neill \cite{oneill}.  
\begin{theorem}[Condition for null curves to be perturbed into timelike curves]
\label{pert_timelike}
Let $S$ be a spacelike submanifold of a spacetime $M$ and $\alpha:[0,1]\rightarrow M$ a causal curve with $\alpha(0)\in S$. Then there exists a timelike curve $\beta$ with $\beta(0) \in S$ and $\beta(1) = \alpha(1)$ unless $\alpha$ coincides, up to reparemtrization, with a null geodesic exiting orthogonally from $S$ and without any focal points before $\alpha(1).$
\end{theorem}
\begin{proposition}[Compactness of $\overline{\Sigma}_-$]
\label{compactness}
    Let $(M,\,g,\,e^{-V}\mathrm{Vol}_M)$ be an $n+1$-dimensional weighted spacetime satisfying the synthetic $N$-null energy condition, which is past reflecting and null geodesically complete. Suppose that $M$ contains a synthetically asymptotically regular hypersurface $\Sigma$ with a piercing. Then for any enclosing surface $S$ in $\Sigma$, we have that $\overline{\Sigma}_-:= \Sigma \,\cup\, S$ is compact. 
    \begin{proof} 
    Consider the set $E^+(S):=J^+(S)\setminus I^+(S)$: by theorem \ref{pert_timelike}, it is composed by segments of null geodesics orthogonal to $S$ that do not contain focal points. Hence the subset $\mathcal{H}^+\subseteq E^+(S)$ of points contained on future directed null geodesics with $\gamma(0) \in S, \gamma'(0) = K_-(\gamma(0))$ is relatively compact, since it is contained in the image of the compact set $S\times [0, t_0]$ via the map $(x,t)\mapsto \exp_x(t\tilde{K}_-(x)).$ Here we are using that this map is indeed well defined since by null geodesic completeness we have that $t\tilde{K}_-(x)$ belongs to the domain of the exponential map for any $t\in \R.$
    \vskip 5pt
    From now on, the proof is the same as in \cite{ces_ming} and \cite{schinnerl_steinbauer}. Define $T=\partial I^+(\Sigma_+)\setminus \Sigma_+$ and recall that, fixing a vector field $X$ that induces a piercing of $\Sigma$, we obtain a continuous map $\rho_X: M\rightarrow \Sigma$ that sends a point $p\in M$ to the unique intersection point between the flow curve of $X$ passing through $p$ and $\Sigma.$ It is proven in O'Neill \cite{oneill}, Proposition 14.31, that $\rho_X$ is an open map. The proof will be now subdivided in two steps: 
    \[(1)\; \mathrm{Compactness\;of\;} T; \qquad \qquad (2)\; \overline{\Sigma}_-=\rho_X(T). \]
    \textbf{1.} We start by proving that $T$ is achronal. Since achronality passes to subsets, it suffices to show that $\partial I^+(\Sigma_+)$ is achronal. Let $x\in \partial I^+(\Sigma_+)$ and $y\in I^+(x).$ We need to show that $y\notin \partial I^+(\Sigma_+).$ By openness of chronological past, there exists a neighborhood of $x$ that is contained in $I^-(y).$ Since $x$ belongs to the boundary of $I^+(\Sigma_+),$ this neighborhood has nonempty intersection with $I^+(\Sigma_+)$ and therefore it holds $y\in I^+(I^+(\Sigma_+)) = I^+(\Sigma_+).$ By openness of chronological futures this implies $y\notin \partial I^+(\Sigma)$, as we wanted. \vskip 5pt
    Note that by definition $T$ is a closed set and hence to prove compactness it suffices to show that it is contained in the relatively compact set $\mathcal{H}^+.$ Assume by contradiction that there exists a point $q\in T\setminus \mathcal{H}^+$. \vskip 5pt
    \noindent\emph{Claim:} there exists a succession $(q_k) \subseteq I^+(S)$ with $q_k\to q.$ \\
    \emph{Proof of claim:} we can take a a succession of points $q_k\in I^+(q)$ that converges to $q,$ so that $q_k\in I^+(\Sigma_+).$ Denote with $\alpha$ the maximal integral curve of $X$ passing by $q.$ By definition of a piercing, $\alpha$ must intersect $\Sigma$ exactly once. By acausality of $\Sigma$ this intersection point is $q$ or it must lie on the past of $q$: in the first case, since $q\notin \overline{\Sigma}_+$ we get $q\in \Sigma_-;$ in the second, since $q\notin I^+(\overline{\Sigma}_+)$ we get $q\in I^+(\Sigma_-).$ Hence in any case $q\in I^+(\Sigma_-).$ Since $q_k$ lies on the future of $q$, this implies that the past of $q_k$ has nonempty intersection with $\Sigma_-.$ Moreover it also intersects $\Sigma_+$ since $q_k\in I^+(\Sigma_+).$ However the set $I^-(q_k)\,\cap\,\Sigma$ coincides with the image of $I^-(q_k)$ via $\rho_X$ and is therefore is connected, hence it must intersect $S.$ Then $S\,\cap\, I^-(q_k) \ne \emptyset \Rightarrow q_k\in I^+(S),$ which proves the claim. \vskip 5pt
    \indent 
    Once again using Theorem 1 of Nomizu and Ozeki's paper \cite{compl_riem}, we can fix a complete background Riemannian metric $h$ on $M$. Consider a sequence of future directed, future inextendible timelike curves $\sigma_k : [0,+\infty)\rightarrow M$ starting from $S$ that are parametrized by $h$-arclength and satisfy $\sigma_k(t_k) = q_k$ for some $t_k >0.$ By compactness of $S$, up to passing to a subsequence we can suppose that $\sigma_k$ converge $h$-uniformly on compact subsets to a curve $\sigma : [0,+\infty)\rightarrow M.$ Moreover by the Limit Curve theorem (see for example Minguzzi \cite{ming}, Theorem 2.14) we can also suppose that $\sigma$ is also a future directed, future inextendible cuasal curve.
    \vskip 5pt
    \noindent\emph{Claim:} The sequence $(t_k)$ is unbounded. \\
    \emph{Proof of claim:} Suppose by contradiction that $(t_k)$ is bounded. Again up to passing to a subsequence we can suppose $t_k\to t_0$ and therefore $\sigma(t_0) = q.$ By achronality of $T$ the curve $\sigma$ cannot be perturbed to a timelike curve, hence once again by theorem \ref{pert_timelike} it must coincide, up to reparametrization, with a null geodesic $\eta$ exiting orthogonally from $S$ that doesn't contain any focal points before $q = \eta(b)$. The initial velocity $\eta'(0)$ can't be proportional to $K_-(\eta(0)),$ otherwise we would have $q\in \mathcal{H}^+$. Therefore $\eta'(0)$ must be proportional to $K_+(\eta(0)).$ However, this implies that $\eta$ must enter $I^+(\Sigma_+)$ at some time $b'<b$ and therefore $q\in I^+(\Sigma_+),$ which contradicts $q\in T.$ This proves that $t_k$ must be unbounded and up to a subsequence we can assume that $t_k\to +\infty.$ \vskip 5pt
    
    Note now that since $\sigma$ is defined on the whole ray $[0,+\infty)$, it can't be contained on $\partial I^+(S),$ otherwise it would coincide up to reparametrization with a null geodesic exiting orthogonally from $S$ in direction $K_-$ and we proved above that all those curves enter $I^+(S)$ in finite time. Then, there exists $s>0$ such that $\sigma(s) \in I^+(S).$ Now let $U$ be an open neighborhood of $\sigma(s)$ that is contained in $I^+(S)$ and $r$ be a point in $I^-(\sigma(s), U).$ Since $\sigma_k(s) \to \sigma(s)$ and $t_k\to +\infty$, we have $\sigma_k(s) \in U$ and $\sigma_k(s) \in I^-(\sigma_k(t_k))$ for $k$ sufficiently big. Therefore eventually it holds
    $r \ll \sigma_k(s) \le \sigma_k(t_k) = q_k.$ Passing to the limit, we obtain $q\in \overline{I^+(r)}.$
    By past reflectivity of $M$,  we obtain $r\in \overline{I^-(q)}$, from which follows $r\in \overline{I^-(q)}\,\cap\, I^+(S)$, where we have used that $r\in I^-(\sigma(s), U)\subseteq U\subseteq I^+(S).$ This implies that there exists a succession of points in $I^-(q)$ that approach $r\in I^+(S)$. By openness of $I^+(S)$ this implies that $I^-(q)\,\cap\, S\ne \emptyset$ and therefore $q \in I^+(S) \subseteq I^+(\Sigma_+).$ This contradicts $q\in T$ and proves compactness of $T$. 
    \vskip 5pt
    
    \noindent \textbf{2.} We conclude the proof by showing that $\overline{\Sigma}_-=\rho_X(T).$ We prove separately the two opposite inclusions. 
    \vskip 5pt
     
     We start by proving $\rho_X(T)\subseteq \overline{\Sigma}_-.$ 
     Suppose there exists a point $q\in \rho_X(T)\setminus \overline{\Sigma}_- \subseteq \Sigma\setminus \overline{\Sigma}_- = \Sigma_+.$ Then there is a point $q'\in T$ which lies on the flow curve $\eta$ of $X$ passing through $q\in \Sigma_+.$ Since $X$ is future directed, this implies either $q'\in I^+(\Sigma_+)$ or $q'\in I^-(\Sigma_-).$ The first case contradicts $q'\in T$ since $T$ is contained in $\partial I^+(\Sigma_+)$ which is disjoint from $I^+(\Sigma_+);$ the second is also absurd since $q'\in T\subseteq J^+(\Sigma_+)$ and therefore joining a future-directed causal curve from $\Sigma_+$ to $q'$ with $\eta$ we would get a causal curve intersecting $\Sigma_+$ twice, which contradicts acausality of $\Sigma.$\vskip 5pt
    \indent For the opposite inclusion, we start by noting that since $S \subset T,$ we have $S=\rho_X(S) \subseteq \rho_X(T)$. Therefore, we obtain $\overline{\Sigma}_-\setminus \rho_X(T) \subseteq \Sigma_-.$ Assume by contradiction that $\overline{\Sigma}_-\setminus\rho_X(T) \ne \emptyset$. Then there exists a point $p\in \partial_\Sigma \rho_X(T) \,\cap\, \Sigma_-,$ where with $\partial_\Sigma$ we denote the topological boundary of a subset of $\Sigma,$ endowed with the subspace topology. By compactness of $T$, we have $p\in \rho_X(T)$ and therefore there exists a point $p' \in T$ such that $p = \rho_X(p').$ Moreover $p' \notin S$ since otherwise we would have $p  = p' \in S\,\cap\, \Sigma_- = \emptyset.$ 
    
    Now note that the edge of $T$ is contained in $S$. Indeed, if $x\in T\setminus S$, we can choose an open neighborhood $U$ of $x$ in $M$ that doesn't intersect $S$. It is easy to show that 
    \[ I^{+}(x,U) \subseteq I^+(\Sigma), \qquad I^{-}(x, U) \subseteq M \setminus \overline{I^{+}(\Sigma)}\]
    (see Galloway \cite{gall_due}, Lemma 3.2). Therefore any curve in $U$ connecting $I^{-}(x,U)$ and $I^{+}(x,U)$ must intersect $T$, which proves that $x \notin $ edge$(T)$.
    By Proposition \ref{achronal_edge} this implies that $T\setminus S$ is a topological hypersurface and we can choose an open neighborhood $V_0$ of $p'$ in $M$ that  satisfies $\rho_X(V_0) \subseteq \Sigma_-$, $V_0\,\cap\,\Sigma = \emptyset$ and such that $V_0\,\cap\, T$ is open in $T\setminus S$ and homeomorphic to an open subset of $\mathbb{R}^n.$
    
    If we denote by $\Psi$ the flow of $X,$ we can choose an open neighborhood $U_0$ of $p'$ in $M$ and a real number $\epl > 0$ that satisfy $\Psi(U_0\times(-\epl, \epl))\subseteq V_0$. Now the map $\Psi_0=\restr{\Psi}{U_0\,\cap\,T \times (-\epl, \epl)}$ is injective by achronality of $T$. Hence, setting $W := \mathrm{Im} \Psi_0,$ by invariance of domain we have that $\Psi_0$ is a homeomorphism onto $W$ and $W$ is open. Now we can finally derive a contradiction:  we have $p\in \rho_X(U_0\,\cap\, T) = \rho_X(W)$, which is open by openness of $\rho_X$, and therefore $p\in \mathrm{int}(\rho_X(T))$, which contradicts $p\in \partial_\Sigma \rho_X(T).$ 
    \end{proof}
\end{proposition}
\subsection{Topological aspects}
We can finally finish the proof using the same topological argument used by Schinnerl and Steinbauer in \cite{schinnerl_steinbauer}. The only difference is that we need to apply our local to global result proved in theorem \ref{local_to_global} to prove that the synthetic NEC passes to covering spacetimes, which is instead immediate in the classical case, since the Ricci curvature is a local object. We include the full proof of this step to highlight where and why we need the local to global property.
\vskip 5pt 
We start by fixing some notation. Denote with $j$ the inclusion of $S$ in $M$ and let $\Phi: \tilde{M} \rightarrow M$ be a connected smooth covering of $M$ that satisfies $\Phi_\#(\pi_1(\tilde{M})) = j_\#(\pi_1(S)).$ Moreover let $m$ be the inclusion of $\Sigma$ in $M.$ Since the piercing vector field $X$ is nowhere vanishing, we can normalize it by fixing a complete background Riemannian metric. This produces a complete vector field whose flow curves intersect each exactly once $\Sigma.$ His flow maps gives a diffeomorphism $M\rightarrow \Sigma \times \R$ and therefore the induced map $m_\# : \pi_1(\Sigma)\rightarrow \pi_1(M)$ is an isomorphism. 
We now prove that $\tilde{\Sigma}:=\Phi^{-1}(\Sigma)$ is connected. Since the restriction $\restr{\Phi}{\tilde{\Sigma}}$ is still a covering, it suffices to show that any two points in the preimage of a point $x\in\Sigma$ lie in the same connected component of $\tilde{\Sigma}$. If $\tilde{x}_1,\tilde{x}_2\in \Phi^{-1}(x)$ we can join them with a path $\tilde{\gamma}$ in $\tilde{M}$, which projects to a loop $\gamma$ on $x$ in $M$. Since $\pi_1(M)\cong \pi_1(\Sigma)$ the loop $\gamma$ is homotopic to a loop in $\Sigma,$ which in turn lifts to a path in $\tilde{\Sigma}$ connecting $\tilde{x}_1$ and $\tilde{x}_2$. 

\begin{lemma}
    The covering $\restr{\Phi}{\tilde{\Sigma}}$ is trivial.
    \begin{proof}
    We start by noting that since $S$ separates $\Sigma$ in two connected components there is an embedding $F: S\times (-1,1) \rightarrow \Sigma$ with $F(S\times\{0\}) = S$ and $F(S\times (0,1))\subseteq \Sigma_-.$ Let $U_F^- := F(S\times (0,1))$ and $W:=U_F^- \cup S\cup \Sigma_+.$ We have that $W$ retracts by deformation on $\overline{\Sigma}_+$, hence $\pi_1(W) \cong \pi_1(\overline{\Sigma}_+)$. Now consider a connected component $\tilde{W}$ of $\Phi^{-1}(W).$ The restriction $\restr{\Phi}{\tilde{W}} : \tilde{W}\rightarrow W$ is still a covering map; we prove now that it is actually a diffeomorphism. Since it is a surjective local diffeomorphism, it suffices to prove that it is injective.  
    \vskip 5pt
    Let $\tilde{p},\tilde{q}\in \tilde{W}$ such that $\Phi(\tilde{p})=\Phi(\tilde{q}) =:p\in W.$ By connectedness of $\tilde{W}$ there exists a path $\tilde{\alpha} : [0,1] \rightarrow \tilde{W}$ with $\tilde{\alpha}(0)=\tilde{p},$ $\tilde{\alpha}(1) = \tilde{q},$ whose projection $\alpha:=\Phi\circ\tilde{\alpha}$ is a loop in $W$. Since $\pi_1(W)\cong \pi_1(\overline{\Sigma}_+)\cong \pi_1(S)$ (with isomorphisms induced by the inclusions), $\alpha$ is homotopic to a loop in $S.$ However, since $\Phi_\#(\pi_1(\tilde{M})) = j_\#(\pi_1(S))$ we have that $\alpha$ actually lifts to a closed loop on $\tilde{M}.$ By uniqueness of the lift, given a fixed base point, this means that $\tilde{\alpha}$ is a loop, i.e., $\tilde{p}=\tilde{q}.$
    \vskip 5pt
    Suppose now by contradiction that $\restr{\Phi}{\Sigma}$ is a non trivial covering. Then, since $S$ is contained in $W$, the preimage $\Phi^{-1}(S)$ has more than one connected components that all disconnect $\tilde{\Sigma}$ and are diffeomorphic to $S$ via $\Phi.$ Let $\tilde{S}_1$ and $\tilde{S}_2$ be two such components and let $\tilde{W}_1$, $\tilde{W}_2$ be the respective connected components of $\Phi^{-1}(W)$ that contain them. Each $\tilde{W}_i$ contains a copy of $\overline{\Sigma}_+,$ that we denote with $\tilde{C}_i$. For $i=1,2$, we have that $\tilde{C}_i$ is non-compact and closed in $\tilde{W}_i.$ Define $\Sigma^{(1)}_-:=\tilde{\Sigma}\setminus \tilde{C}_1.$ All the hypothesis on $M,\Sigma, S, \Sigma_-$ are also satisfied respectively by $\tilde{M}, \tilde{\Sigma}, \tilde{S}_1$ and $\Sigma^{(1)}_-$; indeed, the synthetic $N$-null energy condition lifts thanks to Corollary \ref{covering}, past reflectivity is assumed for $\tilde{M}$ in the hypothesis of the synthetic Gannon--Lee theorem and all the other assumptions are local. Thus, we can apply Proposition \ref{compactness} to obtain that $\Sigma^{(1)}_-\cup\tilde{S}_1$ is compact. However, since $\tilde{W}_1\,\cap\,\tilde{W}_2=\emptyset,$ we have that $$\tilde{C}_2\subseteq \tilde{W}_2\subseteq \tilde{\Sigma}\setminus \tilde{W}_1\subseteq \tilde{\Sigma}\setminus \tilde{C}_1\subseteq \tilde{S}_1\cup \Sigma^{(1)}_-.$$
    Therefore the compact set $\tilde{S}_1\cup \Sigma^{(1)}_-$ contains the closed non-compact subset $\tilde{C}_2,$ which is a contradiction.
    \end{proof}
\end{lemma}
Now the conclusion of the proof of theorem \ref{synth_gan} follows readily. We have the following commutative diagrams:
\[\begin{tikzcd}
	& {\tilde{\Sigma}} & {\tilde{M}} &&& {\pi_1(\tilde{\Sigma})} & {\pi_1(\tilde{M})} \\
	S & \Sigma & M && {\pi_1(S)} & {\pi_1(\Sigma)} & {\pi_1(M)} \\
	&&&&&& {j_\#(\pi_1(S))}
	\arrow["\tilde{m}", hook, from=1-2, to=1-3]
	\arrow["{\restr{\Phi}{\Sigma}}"', from=1-2, to=2-2]
    \arrow[from=2-2, to=1-2]
	\arrow["\Phi", from=1-3, to=2-3]
	\arrow[from=1-6, to=1-7]
    \arrow[from=2-6, to=1-6]
	\arrow[from=1-6, to=2-6]
	\arrow[from=1-7, to=2-7]
	\arrow["{\Phi_\#}", bend left=50, two heads, from=1-7, to=3-7]
	\arrow["i", hook, from=2-1, to=2-2]
	\arrow["j"', bend right=30, hook, from=2-1, to=2-3]
	\arrow["m", hook, from=2-2, to=2-3]
	\arrow["{i_\#}", from=2-5, to=2-6]
	\arrow["{j_\#}"', from=2-5, to=3-7]
	\arrow["{m_\#}", from=2-6, to=2-7]
    \arrow[from=2-7, to=2-6]
	\arrow[hook, from=3-7, to=2-7]
\end{tikzcd}\]
Let $y\in \pi_1(\Sigma).$ By commutativity we have that 
$$m_\#(y) = \left(\Phi \circ \tilde{m} \circ \restr{\Phi}{\Sigma}^{-1}\right)_\#(y)= \Phi_\# \left(\tilde{m}\circ\restr{\Phi}{\Sigma}^{-1}\right)_\#(y) \in j_\#(\pi_1(S)).$$ 
Therefore, there exists an element $x\in \pi_1(S)$ such that $m_\#(y) = j_\#(x).$ Since $j = m\circ i$, we obtain $$m_\#(y) = m_\#(i_\#(x)).$$ By injectivity of $m_\#$ this implies $y = i_\#(x)$, which proves surjectivity of $i_\#$ since $y$ is arbitrary.

\bibliographystyle{alpha}
\bibliography{biblio}

\end{document}

%% file: preamble.tex
\usepackage{vmargin}
\setmarginsrb{35mm}{25mm}{35 mm}{25mm}{0pt}{0mm}{0pt}{0mm}
\usepackage{amssymb}
\usepackage{amsmath}
\usepackage{amsthm}
\usepackage{pgfplots} 
\pgfplotsset{compat=1.9} 
\usepackage{graphicx}
\usepackage[utf8]{inputenc}
\usepackage{tikz}
\usepackage{tikz-cd}
\usepackage{bbm}
\usepackage{subcaption}
\usepackage[boxruled]{algorithm2e}
\usepackage{mathtools}
\usepackage{lipsum}
\usepackage[title,titletoc]{appendix}
\usepackage{booktabs}
\usepackage{here}
\usepackage{natbib,enumerate} 
\usepackage[]{hyperref}
\usepackage{theoremref}
\usepackage{halloweenmath}
\usepackage{knitting}
\usepackage{esint}

\renewcommand{\phi}{\varphi}

\newtheorem{theorem}{Theorem}[section]
\newtheorem{proposition}[theorem]{Proposition}
\newtheorem{lemma}[theorem]{Lemma}
\newtheorem{corollary}[theorem]{Corollary}
\theoremstyle{remark}
\newtheorem{remark}[theorem]{Remark}
\theoremstyle{definition}
\newtheorem{definition}[theorem]{Definition}
\newcommand{\N}{\mathbb{N}}
\newcommand{\R}{\mathbb{R}}

\newcommand{\epl}{\varepsilon}
\newcommand\restr[2]{{
  \left.\kern-\nulldelimiterspace 
  #1 
  \littletaller 
  \right|_{#2} 
  }}

\newcommand{\littletaller}{\mathchoice{\vphantom{\big|}}{}{}{}}